\pdfoutput=1
\documentclass[journal,10pt]{IEEEtran}
\usepackage{cite}
\usepackage{indentfirst}
\usepackage{graphicx}
\usepackage{changepage}
\usepackage{stfloats}
\usepackage{amsfonts,amssymb}
\usepackage{bm}
\usepackage{algorithm}
\usepackage{algorithmic}
\usepackage{amsmath}
\usepackage{amsmath,amssymb,amsfonts}
\usepackage{times}
\usepackage{url}
\usepackage{cite}
\usepackage{textcomp}
\usepackage{xcolor}
\usepackage{color}
\usepackage{setspace}
\usepackage{enumitem}
\usepackage{float}
\usepackage{diagbox}
\usepackage[T1]{fontenc}
\usepackage[utf8]{inputenc}
\usepackage{authblk}
\usepackage{threeparttable}
\usepackage{booktabs}
\usepackage{footnote}
\usepackage{graphicx}
\usepackage{pifont}
\usepackage{subfigure}
\usepackage{mathrsfs}
\allowdisplaybreaks[4]

\newenvironment{proof}{{\indent  \indent \it Proof:}}{\hfill $\blacksquare$}

\begin{document}
\title{Intelligent Surface Empowered Sensing and Communication: A Novel Mutual Assistance Design}
		
\author{
	Kaitao Meng, Qingqing Wu, Wen Chen, Enrico Paolini, and Elisabetta Matricardi
	\thanks{K. Meng is with the State Key Laboratory of Internet of Things for Smart City, University of Macau, Macau, 999078, China (email: kaitaomeng@um.edu.mo). Q. Wu and W. Chen are with the Department of Electronic Engineering, Shanghai Jiao Tong University, 200240, China (emails: \{qingqingwu, wenchen\}@sjtu.edu.cn). E. Paolini and E. Matricardi are with the Wireless Communications Laboratory, CNIT, DEI, University of Bologna, 40136 Bologna, Italy (emails: \{e.paolini,  elisabett.matricard3\}@unibo.it)}  
}
\maketitle

\begin{abstract}
Integrated sensing and communication (ISAC) is a promising paradigm to provide both sensing and communication (S\&C) services in vehicular networks. However, the power of echo signals reflected from vehicles may be too weak to be used for future precise positioning, due to the practically small radar cross section of vehicles with random reflection/scattering coefficient. To tackle this issue, we propose a novel mutual assistance scheme for intelligent surface-mounted vehicles, where S\&C are innovatively designed to assist each other for achieving an efficient win-win integration, i.e., sensing-assisted phase shift design and communication-assisted high-precision sensing. Specifically, we first derive closed-form expressions of the expected echo power and achievable rate under uncertain angle information. Then, the communication rate is maximized while satisfying sensing requirements, which is proved to be a monotonic optimization problem on time allocation. Furthermore, we unveil the feasible condition of the problem and propose a polyblock-based optimal algorithm. Simulation results validate that the performance trade-off bound of S\&C is significantly enlarged by the novel design exploiting mutual assistance in intelligent surface-aided vehicular networks.
\end{abstract}   

\vspace{-1mm}
\begin{IEEEkeywords}
	\vspace{-0.5mm}
	Mutual assistance, intelligent surface, integrated sensing and communication, vehicular communication.
\end{IEEEkeywords}
\newtheorem{thm}{\bf Lemma}
\newtheorem{remark}{\bf Remark}
\newtheorem{Pro}{\bf Proposition}
\newtheorem{theorem}{\bf Theorem}
\newtheorem{Assum}{\bf Assumption}
\newtheorem{Cor}{\bf Corollary}

\vspace{-2.3mm}
\section{Introduction}
\vspace{-1mm}
Vehicle-to-everything (V2X) communications have attracted considerable attention due to the growing popularity of vehicular applications, such as autonomous driving and intelligent transportation systems \cite{Challenges2021Gyawali}. Driven by these applications, besides high-quality communication, the requirements of vehicle localization accuracy and latency have become more stringent. By far, sensing and communication (S\&C) functionalities in the same frequency are usually designed separately, which inevitably leads to severe mutual interference between S\&C \cite{Zheng2019RadarCommunication}. To fully exploit the potential of the limited wireless resources, integrated sensing and communication (ISAC) techniques have been proposed to simultaneously convey communication information to receivers and extract target information from scattered echoes, thereby providing higher spectral efficiency with lower hardware cost \cite{Liu2022Survey}. 

Most existing works regarding ISAC focus on the investigation of trade-offs between S\&C, such as resource allocation and beampattern design \cite{meng2022throughput}. In addition to exploring the integration gain of S\&C, the interplay between S\&C offers the potential to achieve coordination gain in vehicular networks. Some useful sensing-assisted communication schemes are proposed to reduce the overhead of beam training \cite{liu2020radar}. For instance, \cite{liu2020radar} proposed a  predictive beamforming scheme to improve the sensing accuracy while guaranteeing communication performance by tracking the vehicle based on echo signals. However, the radar cross section of vehicles is generally small and the corresponding reflection/scattering coefficient is random. As a result, the power of echo signals reflected from vehicles may be too weak and fluctuating to be utilized for high-precision and high-reliability positioning, especially for the limited transmit power of road site units (RSUs). Moreover, dedicated information signals generally lead to interference to other unintended users in traditional wireless communications. Nevertheless, from a sensing perspective, the echo signal generated by this interference provides new opportunities to improve localization performance. This thus motivates us to design communication-assisted sensing schemes with the exploitation of interference.
 
 
Intelligent reflecting/refracting surfaces (IRSs) \cite{Simultaneously2022Mu} have been proposed to reconfigure wireless channels, and thus enhance the S\&C performance. For example, in \cite{singh2022visible, Gu2022Intelligent}, IRSs were deployed on the buildings/roadside to enhance the communication performance and service coverage in vehicular networks. Different from traditional schemes deploying IRSs at fixed locations, we propose a novel design to overcome the localization challenges caused by the weak/fluctuating power of echoes and realize mutual assistance between S\&C, i.e., deploying IRSs on the surface of vehicles. Specifically, the signals refracted through the IRSs are focused on the users inside vehicles to improve communication performance while the echo signals reflected by the IRS are boosted at the RSU receivers for enhancing sensing performance and reducing potential interference to other wireless systems. In this case, the sensing results can be further utilized to design IRS phase shifts for communication; while the interference caused by the information signal to unintended users can also be exploited for vehicle detection and tracking. Therefore, S\&C are able to assist each other for achieving an efficient win-win integration, i.e., sensing-assisted phase shift design and communication-assisted high-precision sensing. The main contributions of this paper are summarized as follows: 
\begin{itemize}[leftmargin=*]
	\item We propose a novel design exploiting mutual assistance between S\&C. The achievable rate is improved by utilizing sensing results while sensing performance is enhanced by exploiting IRS's reflected information signals.
	\item We derive closed-form expressions of the echo power and achievable rate under uncertain angle information, based on which, the considered problem is transformed and optimally solved through the proposed polyblock-based algorithm.
	\item Simulation results validate that the S\&C performance bound is greatly enlarged by the novel design exploiting mutual assistance in intelligent surface-aided vehicular networks.
\end{itemize}

\begin{figure}[t]
	\centering
	\includegraphics[width=7.8cm]{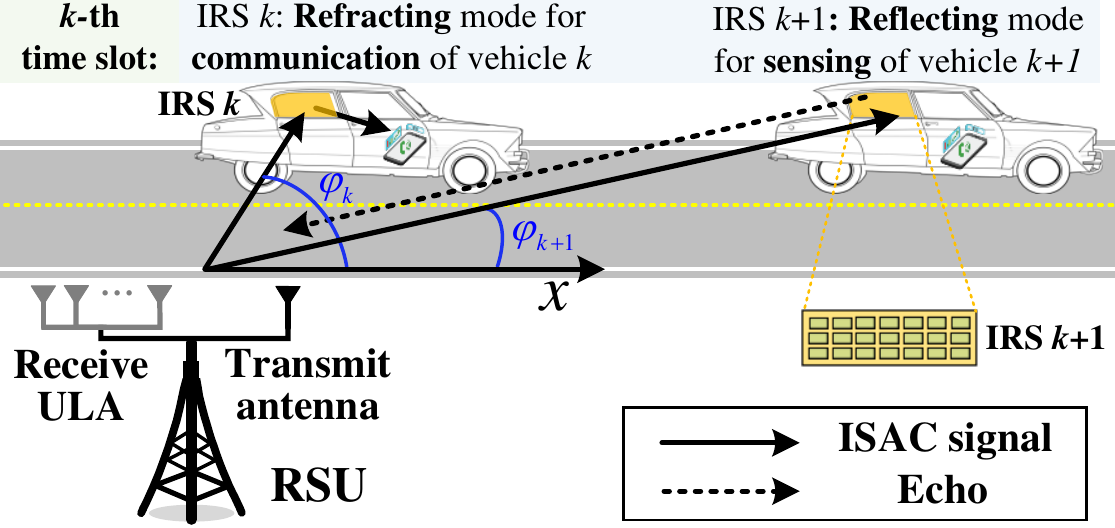}
	\vspace{-3.5mm}
	\caption{Mutually assisted S\&C in vehicular networks with IRSs.}
	\label{figure1a}
	\vspace{-4.5mm}
\end{figure}

\vspace{-4mm}
\section{System Model and Problem Formulation}
\vspace{-0.5mm}
\label{AOAEstimationError}
As shown in Fig.~\ref{figure1a}, we consider an IRS-aided ISAC system to achieve mutual assistance between S\&C, where one RSU sends independent signals to $K$ IRS-mounted mobile vehicles (indexed by $k \in {\cal{K}} = \{1,\cdots,K\}$) in a time-division manner and analyzes the echo signals reflected from IRSs simultaneously. The RSU employs one transmit antenna and one general uniform linear array (ULA) with $M_r$ receive antennas along the $x$-axis. Vehicles are assumed to drive along a straight road parallel to the $x$-axis.\footnote{It can be extended to roads with arbitrary geometry by establishing the relationship between the vehicle state estimation and the road geometry.} An IRS with reflecting and refracting modes is deployed on the surface of each vehicle to replace part of its glass windows or metal panels. The ULA with half-wavelength antenna spacing is adopted at the IRS,\footnote{A uniform planar array (UPA) can be readily extended by adding another receive ULA at the RSU to estimate the elevation angle to the IRSs.} and the number of IRS elements equipped on each vehicle is $L$, where the element is indexed by $l \in {\cal{L}} = \{ 1,\cdots,L\}$. 

\vspace{-2.5mm}
\subsection{Proposed Framework for Mutually Assisted S\&C}
\label{FrameworkMutualBenefit}
\vspace{-1mm}
As shown in Fig.~\ref{figure1b}, the whole service period $T$ is divided into multiple frames with equal length $\Delta T$, indexed by $n \in \{1,\cdots, N \}$, where $N = {{T}/{\Delta T}}$. As done in \cite{liu2020radar}, the motion parameters are assumed to keep constant within each frame. Each frame is further divided into $K+1$ time slots, indexed by $i \in \{0\} \cup {\cal{K}}$, and the ratio of the $i$th time slot to the current frame is denoted by $\eta_i$, with $\sum\nolimits_{i = 0}^K \eta_i = 1$.\footnote{Idle intervals can be added between adjacent time slots to improve the robustness of time synchronization between vehicles.} An example of the synergy between S\&C is provided as follows:
\begin{itemize}[leftmargin=*]
	\item {\textit{Sensing-assisted communication}}: At the $k$th time slot, the RSU sends signal $s_k(t)$ to vehicle $k$. IRS $k$ is in the \textit{refracting} mode and its phase shifts are designed for communication based on the sensing results at time slot $k-1$.
	\item {\textit{Communication-assisted sensing}}: Signal $s_k(t)$ incident on IRS $k+1$ at the $k$th time slot is generally considered useless, but can be utilized to enhance echo signals and help localize vehicle $k+1$ by setting IRS $k+1$ into the \textit{reflecting} mode.
\end{itemize} 
As shown in Fig.~\ref{figure1b}, the first and last time slots are respectively designed for dedicated sensing of vehicle 1 and communication of vehicle $K$. Overall, the S\&C tasks of vehicle $k$ are executed sequentially at the $(k \!- \! 1)$th and $k$th time slots, while the communication task of vehicle $k$ and the sensing task of vehicle $k \!+ \! 1$ are executed simultaneously at the $k$th time slot. IRS $k$ is set to be at the refracting/reflecting mode for communication/sensing at the $k$th/$(k-1)$th time slot, and be switched off at other time slots. Then, in the following frames, the IRSs are set to reflecting or refracting mode in the same scheduling order. Note that the refracting and reflecting coefficient matrices of IRS $k$ at the $k$th and $(k-1)$th time slots are respectively given by ${\bm{\Theta}}^{\mathrm{T}}_{k,n} = {\rm{diag}}(e^{j \theta^{\mathrm{T}}_{k,1,n}}, ... , e^{j \theta^{\mathrm{T}}_{k,L,n}})$ and ${\bm{\Theta}}^{\mathrm{R}}_{k,n} = {\rm{diag}}(e^{j \theta^{\mathrm{R}}_{k,1,n}}, ... , e^{j \theta^{\mathrm{R}}_{k,L,n}})$, where $\theta^{\mathrm{T}}_{k,l,n}$ and $\theta^{\mathrm{R}}_{k,l,n} \in [0, 2 \pi)$, respectively represent the refraction- and reflection-phase shifts of the $l$th element of IRS $k$. Based on the above design, S\&C interference can be effectively avoided even if the distance between vehicles is small.

As shown in Fig.~\ref{figure1a}, during the $n$th frame, the bearing angle of vehicle $k$ relative to the receive ULA is estimated by super-resolution algorithms like multiple signal classification (MUSIC) \cite{schmidt1986multiple}, denoted by $\varphi_{k,n}$, and these instantaneous channel information is further transmitted to vehicle $k$ for phase shift design.During the sensing/communication time of each vehicle, the IRS phase shift matrices ${\bm{\Theta}}^{\mathrm{R}}_{k,n}$/${\bm{\Theta}}^{\mathrm{T}}_{k,n}$ are designed based on the instantaneous state estimation/state tracking angles, denoted by $\varphi_{k,n|n-1}$/$\tilde \varphi_{k,n}$. Specifically, the state estimation model of vehicle $k$ can be given by \cite{liu2020radar}
\vspace{-1.5mm}
\begin{equation}\label{StateEsitmation}
	\left\{ \!\!\begin{array}{l}
		\varphi_{k,n|n-1}=\varphi_{k,n-1}+d_{k,n-1}^{-1} v_{k,n-1} \Delta T \sin \left(\varphi_{k,n-1}\right)+\omega_{\varphi} \\
		d_{k,n|n-1}=d_{k,n-1}-v_{k,n-1} \Delta T \cos \left(\varphi_{k,n-1}\right)+\omega_{d} \\
		v_{k,n|n-1}=v_{k,n-1}+\omega_{v}
		\vspace{-2.5mm}
	\end{array}\right. \!\!\!\!\!,
\end{equation}
where $d_{k,n}$ and $v_{k,n}$ respectively represent the distance from the RSU to IRS $k$ and the velocity of vehicle $k$ during the $n$th frame. In (\ref{StateEsitmation}), $\omega_{\varphi}$, $\omega_{d}$, and $\omega_{v}$ denote the estimation noise, and $\sigma^2_{\omega_{\varphi}}$, $\sigma^2_{\omega_{d}}$, and $\sigma^2_{\omega_{v}}$ are the corresponding variances. During the $n$th frame, the measured angle of vehicle $k$ can be modeled as $\hat \varphi_{k,n}=\varphi_{k,n}+{z}_{\varphi_{k,n}}$, where ${z}_{\varphi_{k,n}} \in \mathcal{C} \mathcal{N}(0, {\sigma}^2_{z_{k,n}} )$ is the measurement error. Then, the Kalman filtering is adopted for vehicle tracking. Accordingly, the tracking angle can be expressed as $\tilde \varphi_{k,n} = \varphi_{k,n|n-1} + \frac{\sigma _{{\omega_{\varphi}}} ^2}{\sigma _{{\omega_{\varphi}}} ^2 + \sigma^2_{z_{\varphi_{k,n}}}} (\hat \varphi_{k,n} - \varphi_{k,n|n-1})$, and the variance of the tracking angle can be expressed as
\vspace{-2mm}
\begin{equation}\label{StateTrackingError}
	\sigma _{\tilde \varphi_{k,n}} ^2 = {\sigma _{{\omega_{\varphi}}} ^2} - {\sigma _{{\omega_{\varphi}}} ^4}({\sigma _{{\omega_{\varphi}}} ^2 + \sigma^2_{z_{\varphi_{k,n}}}})^{-1} .
\end{equation}

\begin{figure}[t]
	\centering
	\includegraphics[width=8.7cm]{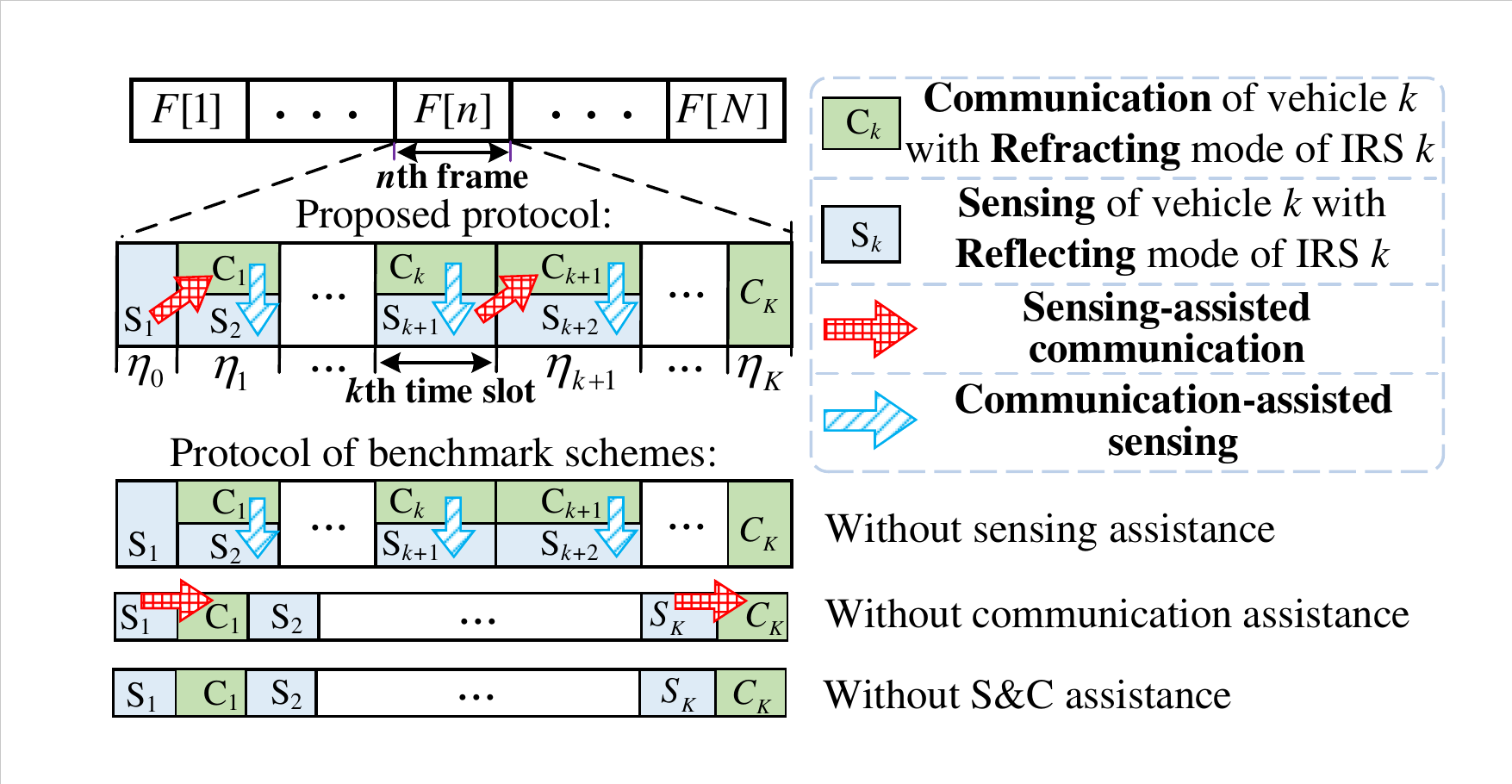}
	\vspace{-3mm}
	\caption{Proposed mutually assisted protocol and benchmark scheme protocols.}
	\label{figure1b}
	\vspace{-5mm}
\end{figure}

\vspace{-4mm}
\subsection{Radar Measurement and Communication Model}
\label{RadarMeasurement}
The channel power gain between the RSU and IRS $k$ can be given by $\beta_{G,k,n} = \beta_0 d_{k,n}^{-2}$, where $\beta_0$ is the channel power at the reference distance 1 meter (m). $\bm{h}^{\mathrm{DL}}_{k,n}  \in \mathbb{C}^{L \times 1}$ and $\bm{H}^{\mathrm{UL}}_{k,n}$ $\in \mathbb{C}^{M_r \times L}$ are respectively the downlink and uplink channel matrices between the RSU and IRS $k$, given by
\vspace{-1.5mm}
\begin{equation}\label{SteeringVector}
	\bm{h}^{\mathrm{DL}}_{k,n}=\sqrt{\beta_{G,k,n}} \bm{a}_{\mathrm{IRS}}\left(-\varphi_{k,n}\right) ,
\end{equation}
\begin{equation}
	\bm{H}^{\mathrm{UL}}_{k,n}=\sqrt{\beta_{G,k,n}}\bm{b}_{\mathrm{RSU}}\left(\varphi_{k,n}\right) \bm{a}_{\mathrm{IRS}}^{T}\left( -\varphi_{k,n}\right), 
	\vspace{-1.5mm}
\end{equation}
where ${\bm{a}}_{\mathrm{IRS}}(-\varphi_{k,n}) = [1, \cdots, e^{ {j \pi(L-1)  \cos(\varphi_{k,n}) }}]^T$ and $ {\bm{b}}_{\mathrm{RSU}}(\varphi_{k,n}) = \left[1, \cdots, e^{ {-j \pi(M_r-1)  \cos(\varphi_{k,n}) }}\right]^T$. The IRS-device link is assumed to be quasi-static during each frame, expressed as ${\bm{h}}_k	= \sqrt{\beta_{h,k}} [1, \cdots, e^{ {-j \pi(L-1) \cos({\varphi}^u_{k,n} )}}]^T$, since IRS $k$ and the communication device inside vehicle $k$ (namely communication device $k$) keep relatively stationary. Here, $\beta_{h,k}$ represents the channel power gain from IRS $k$ to communication device $k$, and ${\varphi}^u_{k,n}$ is the bearing angle of communication device $k$ relative to IRS $k$, which can be obtained based on the device location or existing channel estimation methods \cite{Wu2022IntelligentReflecting}.

At the $(k-1)$th time slot, the signal $s_{k-1}(t)$ incident on IRS $k$ is reflected towards the RSU for positioning vehicle $k$. The echo signals reflected from IRS $k$ can be expressed as $
{\bm{r}}_{k}(t) \!=\!  \sqrt{P_{\mathrm{A}}}e^{j2\pi \mu_{k,n} t}\! {\bm{H}}^{\mathrm{UL}}_{k,n} {\bm{\Theta}}_{k,n}^{\mathrm{R}} {\bm{h}}^{\mathrm{DL}}_{k,n}  s_{k-1,n}(t - \tau_{k,n}) + {\bm{z}}_r(t)$, where $\mu_{k,n}$ and $\tau_{k,n}$ denote Doppler frequency and the round-trip delay of echo signals, $P_{\mathrm{A}}$ denotes the constant transmit power, and ${\bm{z}}_r(t)$ represents the noise at the receive ULAs. Generally, the error of estimated angles is inversely proportional to the echo power at the RSU \cite{Liu2022Survey}, i.e.,
\vspace{-1.5mm}
\begin{equation}\label{ReceivedPowerEquation}
	\sigma_{z_{\varphi_{k,n}}}^{2} \propto ({{\gamma^{\mathrm{S}}_{k,n}{{\sin }^2}(\varphi_{k,n}) }})^{-1}, 
	\vspace{-1.5mm}
\end{equation}
where ${\gamma^{\mathrm{S}}_{k,n}}$ is the signal-to-noise ratio (SNR) at the receive ULA after match-filtering. Due to the impact of the uncertain angle on phase shift design, at the $(k-1)$th time slot, the SNR of the received echos is expressed in expectation form, i.e.,
\vspace{-1.5mm}
\begin{equation}\label{SensingReceivedPower}
	\gamma^{\mathrm{S}}_{k,n} \!=\! \frac{{\eta_{k-1} W P_{\mathrm{A}}}}{\sigma_s^2} \mathbb{E}_{\varphi_{k,n|n-1}} \!\left[ \left| {\bm{v}}_{k,n}^{H} {\bm{H}}^{\mathrm{UL}}_{k,n} {\bm{\Theta}}^{\mathrm{R}}_{k,n} {\bm{h}}^{\mathrm{DL}}_{k,n}\right|^{2} \right]\!.
	\vspace{-1.5mm}
\end{equation}
In (\ref{SensingReceivedPower}), ${\bm{v}}_{k,n}^{H}$ is the receive beamforming vector, $\sigma_s^2$ is the noise power at the receive ULA, $\eta_{k-1} W$ denotes the matched-filtering gain, and ${W}$ is the number of symbols in one frame. 

Notice that non-line-of-sight (NLOS) links between the RSU and in-vehicle communication devices are practically negligible due to severe penetration loss, especially for high-frequency signals. Thus, devices are assumed to receive the signal through RSU-IRS-device links. At the $k$th time slot, the signal received at communication device $k$ is given by
\vspace{-1.5mm}
\begin{equation}
	y_{k,n}(t) =   \sqrt{P_{\mathrm{A}}} {\bm{h}}^T_k{\bm{\Theta}}_{k,n}^{\mathrm{T}}\bm{h}^{\mathrm{DL}}_{k,n} s_{k,n}(t) +    z_k(t),
	\vspace{-1.5mm}
\end{equation}
where ${{z}}_k(t)$ denotes the noise at the receive antennas. The signal-to-interference-plus-noise ratio (SINR) of device $k$ is expressed as $\gamma^{\mathrm{C}}_{k,n}  = {  P_{\mathrm{A}}|{\bm{h}}^T_k{\bm{\Theta}}^{\mathrm{T}}_{k,n} {\bm{h}}^{\mathrm{DL}}_{k,n}|^{2}}/{ \sigma_{c}^{2}}$, and its achievable rate in bps/Hz is given in expectation form since the phase shifts are designed based on the uncertain angle ${\tilde \varphi_{k,n}}$, i.e.,
\vspace{-1.5mm}
\begin{equation}\label{AchievablerateSC}
	R_{k,n} \!=\! \eta_k \mathbb{E}_{\tilde \varphi_{k,n}}\left[ \log _{2}(1+ \gamma^{\mathrm{C}}_{k,n})\right]  \!\le\! \eta_k \log _{2}\!\left(1+ \mathbb{E}_{\tilde \varphi_{k,n}}\left[\gamma^{\mathrm{C}}_{k,n}\right] \!\right) \!.
	\vspace{-1.5mm}
\end{equation}

\vspace{-3mm}
\subsection{Problem Formulation}
\vspace{-1mm}
In this work, the minimum achievable communication rate among devices is maximized subject to their sensing requirements by jointly optimizing IRS phase shifts and time allocation, i.e.,             
\vspace{-2.5mm}
\begin{alignat}{2}
	\label{P1}
	(\rm{P1}): \quad & \begin{array}{*{20}{c}}
		\mathop {\max }\limits_{\{{\bm{\Theta}}^{\mathrm{R}}_{k,n}\}, \{{\bm{\Theta}}^{{\mathrm{T}}}_{k,n}\}, {\bm{\eta}}} \quad   \mathop {\min }\limits_k \  R_{k,n}
	\end{array} & \\ 
	\mbox{s.t.}\quad
	& \theta^{\mathrm{T}}_{k,l,n}, \theta^{\mathrm{R}}_{k,l,n} \in [0, 2 \pi), \forall l \in {\cal{L}}, k \in {\cal{K}}, \tag{\ref{P1}a} \\
	& 	\gamma^{\mathrm{S}}_{k,n} \ge \gamma^{th}, \forall k \in {\cal{K}}, \tag{\ref{P1}b} \\
	&  \sum\nolimits_{k = 0}^K \eta_{k} = 1, \tag{\ref{P1}c} \\
	& \eta_{k} \in [0, 1], \forall k \in \{0,\cdots,K\}, & \tag{\ref{P1}d}
	\vspace{-1.5mm}
\end{alignat}
where ${\bm{\eta}} = [\eta_{0},\cdots,\eta_K]$. Constraint (\ref{P1}b) ensures that the power of echo signals is larger than $\gamma^{th}$, which is set according to the sensing requirement of practical applications. {{From the perspective of robustness, the threshold $\gamma^{th}$ can be set higher based on the potential system error, such as the position deviation of vehicles and the RSU.}} Solving (P1) optimally is non-trivial due to the closely coupled variables and the lack of closed-form objective function.

\section{Proposed Optimal Solution}
In this section, we first derive closed-form expressions of the echo SNR and achievable rate, based on which, the problem is optimally solved by the proposed polyblock-based algorithm. 

\subsection{Closed-form Expression of Achievable Rate}

First, it can be readily proved that the received power at the RSU is maximized when the echo reflected from the IRS is directed toward the RSU. Thus, the phase shift of the $l$th element of IRS $k$ in reflecting mode should be set based on the estimated angle $\varphi_{k,n|n-1}$, i.e.,
\begin{equation}
	\theta^{\mathrm{R}}_{k,l,n}=- 2 \pi(l-1)  \cos( \varphi_{k,n|n-1})+\theta_{0}, 
\end{equation}
where $\theta_{0}$ is the reference phase at the origin of the coordinates. Similarly, the phase shifts of all IRS elements are designed to align the refracting beam towards the communication device. Based on the obtained tracking angle $\tilde \varphi_{k,n}$, the phase shift of the $l$th element of IRS $k$ in refracting mode is given by $
	\theta^{\mathrm{T}}_{k,l,n}=\pi(l-1) (\cos( \varphi^{u}_{k,n})- \cos( \tilde \varphi_{k,n}))+\theta_{0}$.
Then, the passive beamforming gain in (\ref{SensingReceivedPower}) is given by
\vspace{-1mm}
\begin{equation}
	\left|\bm{a}^{T}_{\mathrm{IRS}}\left(-\varphi_{k,n}\right) {\bm{\Theta}} ^{\mathrm{T}}_{k,n}\bm{a}_{\mathrm{IRS}}\left(-\varphi_{k,n}\right)\right|^2 
	\buildrel \Delta \over =  L F_{L}\left( 2\Delta \cos \varphi_{k,n}\right) ,
\end{equation}
where $\Delta \cos \varphi_{k,n} = \cos \left( \varphi_{k,n|n-1} \right) - \cos \left( \varphi_{k,n}  \right)$ and the Fej$\acute{e}$r kernel $F_{L}(x) = \frac{1}{L}\left({\sin \frac{L \pi x}{2 }}/{\sin \frac{\pi x}{2 }}\right)^{2}$ \cite{rust2013convergence}. Similarly, the receive beamforming vector is set based on $\varphi_{k,n|n-1}$, and thus  $|{\bm{v}}_{k,n}^{H} {{\bm{b}}_{\mathrm{RSU}}( \varphi_{k,n})}|^2  = F_{M_r}(\Delta \cos \varphi_{k,n}) $. Based on the above analysis, at the receive ULA, the expected SNR of the echo signals reflected from IRS $k$ can be expressed as
\vspace{-1mm}
\begin{equation}\label{SensingPowerAverage}
	\begin{aligned}
		\gamma^{\mathrm{S}}_{k,n} = & \mathbb{E}_{\varphi_{k,n|n-1}}\left[ F_{L}\left( 2\Delta \cos \varphi_{k,n}\right) F_{M_r}\left(\Delta \cos \varphi_{k,n}\right)\right] \\
		& \times  \eta_{k-1} { W {P_{\mathrm{A}}} {\beta^2_{G,k,n}} L  } \left({\sigma_s^2}\right)^{-1},
	\end{aligned}
\vspace{-1mm}
\end{equation}
where $\beta_{G,k,n} \! \approx \! \beta_0 d_{k,n|n-1}^{-2}$ due to unavailable real distance. However, it is intractable to optimize time resources due to the expectation operation in (\ref{SensingPowerAverage}). To resolve this issue, we present the following proposition. 

\begin{Pro}\label{InftyBand}
	When $L \to \infty$, the expectation of SNR at the receive ULA is approximated by 
	\vspace{-1.5mm}
	\begin{equation}\label{ClosedFormSensingPower}
		\gamma^{\mathrm{S}}_{k,n} \approx \eta_{k-1} \frac{  W{P_{\mathrm{A}}}  \beta_{G,k,n}^2 L  M_r h(\varphi_{k,n|n-1}, \sigma^2_{\omega_\varphi}) }{{\sigma_s^2} } ,
	\end{equation}
	where $h(x, y) \! \buildrel \Delta \over =  \!\frac{1}{{\sqrt {2\pi {y}\sin^2({x})}  }} \sum\limits_{i = -\infty}^{\infty}  \!\left( \! {{e^{ - \frac{{{{2( {i\pi} )}^2}}}{{{y}}}}}  \!+ \! {e^{ - \frac{{{{2\left( {(i + 1)\pi  } - x \right)}^2}}}{{{y}}}}}}  \!\right) \!$.
\end{Pro}
\begin{proof}
	Please refer to Appendix A.
\end{proof}

Proposition \ref{InftyBand} provides a tight approximation for the SNR of echo signals under a sufficiently large number of IRS elements, as shown later in practical setups in Section \ref{SimulationSection}. According to (\ref{ReceivedPowerEquation}), $\sigma_{z_{\varphi_{k,n}}}^{2} \!\!\!=\! {{\sigma _R^2}}/{{\gamma^{\mathrm{S}}_{k,n}{{\sin }^2}(\varphi_{k,n|n-1}) }} \!= \! {A_{\varphi_{k,n}}}/{\eta_{k-1}}$, where $\sigma _R^2$ is the variance parameter of the estimation method, and $A_{\varphi_{k,n}} = \frac{   \sigma_s^2 \sigma _R^2 }{{W {P_{\mathrm{A}}}  \beta_{G,k,n}^2 L M_r h(\varphi_{k,n|n-1}, \sigma^2_{\omega_\varphi}) {{\sin^2 }}\varphi_{k,n|n-1} } }$ represents the angle variance when $\eta_{k-1} = 1$.
In $h(x,y)$, it can be readily proved that the term $ {{e^{ - \frac{{{{2\left( {i\pi } \right)}^2}}}{{y}}}} + {e^{ - \frac{{{{2( {(i + 1)\pi } - \varphi_{n|n-1} )}^2}}}{{{y}}}}}} \ll 1$ when $i \ge 1$ or $i \le -1$ and it is negligible, i.e.,
\vspace{-2mm}
\begin{equation}\label{ApproximationEquation}
	h(x,y) \!\approx\! \left({2\pi{y} \sin^2({x})}\right)^{-\frac{1}{2}}\!\left({1 + {e^{ - \frac{{{{2\left( {\pi  - {x} } \right)}^2}}}{{y}}}}} \! \right)\! \buildrel \Delta \over=\! \tilde h(x, y) .
	\vspace{-1.5mm}
\end{equation} 
Similar to the derivation in Proposition \ref{InftyBand}, the achievable rate of device $k$ can be approximated by 
\vspace{-1.5mm}
\begin{equation}
	\tilde R_{k,n} = \eta_{k} \log_2\left( 1 +  C_k  \tilde h(\varphi_{k,n|n-1}, \sigma^2_{\tilde \varphi_{k,n}}) \right),
	\vspace{-1.5mm}
\end{equation}
where $C_k = \frac{ 2 P_{\mathrm{A}}{\beta_{G,k,n}} {\beta_{h,k}}    L }{\sigma_{c}^{2}} $ and $ \sigma^2_{\tilde \varphi_{k,n}} =  \frac{\sigma _{{\omega_{\varphi}}} ^2 A_{\varphi_{k,n}}}{\sigma _{{\omega_{\varphi}}} ^2 {\eta_{k-1}} + {A_{\varphi_{k,n}}}}$. Then, (P1) can be recast in approximate form as  
\vspace{-2mm}
\begin{alignat}{2}
	\label{P2}
	(\rm{P2}): \quad & \begin{array}{*{20}{c}}
		\mathop {\max }\limits_{{\bm{\eta}}} \quad   \mathop {\min }\limits_k \ \tilde R_{k,n}
	\end{array} & \\ 
	\mbox{s.t.}\quad
	& (\ref{P1}b) - (\ref{P1}d). \nonumber
	\vspace{-1mm}
\end{alignat}

\noindent It is difficult to obtain the optimal solution to (P2) due to the non-convex objective function and the closely coupled variables in (\ref{P1}b).

\vspace{-1mm}
\subsection{Proposed Optimal Solution to (P2)}
\vspace{-0.5mm}
In this subsection, we present a feasibility analysis and then propose a polyblock-based algorithm to optimally solve (P2).
\begin{Pro}\label{OptimalFeasibleCheck}
	There is a feasible solution to problem (P2) if and only if 
	\vspace{-2mm}
	\begin{equation}
		\gamma^{th} \le \frac{  W {P_{\mathrm{A}}}  \beta_{G,k,n}^2 L  M_r h(\varphi_{k,n|n-1}, \sigma^2_{\omega_\varphi}) }{{ \sigma_s^2}\left(1 + \sum\nolimits_{k = 2}^{K} \frac{\beta_{G,1,n}h(\varphi_{1,n|n-1}, \sigma^2_{\omega_\varphi})}{{\beta_{G,k,n}h(\varphi_{k,n|n-1}, \sigma^2_{\omega_\varphi})}}\right) }.
		\vspace{-1.5mm}
	\end{equation}
\end{Pro}
\begin{proof}
	It can be readily proved that under the maximum feasible sensing SNR, $\eta_K = 0$, since the $K$th time slot is not designed for sensing. Also, in this case, $\gamma^{\mathrm{S}}_k = \gamma^{\mathrm{S}}_j$, $\forall k \ne j$ at the optimal solution of (P2). Thus, it follows that
	\vspace{-1.5mm}
	\begin{equation}
		{\eta_{k-1}}{\beta_{G,k,n}h(\varphi_{k,n|n-1}, \sigma^2_{\omega_\varphi}\!)} \!=\! {\eta_{k}}{\beta_{G,k+1,n}h(\varphi_{k+1,n|n-1}, \sigma^2_{\omega_\varphi}\!)}.
		\vspace{-1.5mm}
	\end{equation}
	Then, $\sum\nolimits_{k = 0}^{K-1} \eta_k \!=\! \eta_0 \!+\! \sum\nolimits_{k = 2}^{K} \! \frac{{\eta_{0}}{\beta_{G,1,n}h(\varphi_{1,n|n-1}, \sigma^2_{\omega_\varphi})}}{{\beta_{G,k,n}h(\varphi_{k,n|n-1}, \sigma^2_{\omega_\varphi})}} \!=\! 1$, and $\eta_0 \!=\! \frac{1}{1 + \sum\nolimits_{k = 2}^{K}\! \frac{\beta_{G,1,n}h(\varphi_{1,n|n-1}, \sigma^2_{\omega_\varphi})}{{\beta_{G,k,n}h(\varphi_{k,n|n-1}, \sigma^2_{\omega_\varphi})}}}$. Hence, the maximum SNR of echo signals $\gamma^{\mathrm{S}}_1 \!=\! \eta_{0} \frac{  W {P_{\mathrm{A}}}  \beta_{G,1,n}^2 L  M_r h(\varphi_{1,n|n \! -\!1}, \sigma^2_{\omega_\varphi}) }{{ \sigma_s^2} } $.
\end{proof}

Next, we prove that (P2) is essentially a monotonic optimization (MO) problem.

\begin{Pro}\label{VarianceMonotonicity}
	As $\eta_{k-1}$ or $\eta_k$ increases, $R_k$ increases monotonically.
\end{Pro}
\begin{proof}
	Please refer to Appendix B.
\end{proof}

Proposition \ref{VarianceMonotonicity} implies that (P2) is an MO problem with respect to ${\bm{\eta}}$, which can be optimally solved based on the framework of the Polyblock algorithm \cite{zhang2013monotonic}. Specifically, in the $r$th iteration of the proposed algorithm, the vector ${\bm{\eta}}^{(r)} = [\eta_0^{(r)}, \cdots, \eta_K^{(r)}]$ corresponding to the maximum achievable rate $\bar R^{(r)}$ is selected, and then its projection point $\Phi\left({\bm{\eta}}^{(r)}\right) = [\Phi_0\left({\bm{\eta}}^{(r)}\right),\cdots,\Phi_{K}\left({\bm{\eta}}^{(r)}\right)]$ is calculated, where the $k$th coordinate of the projection point can be expressed as
\vspace{-1.5mm}
\begin{equation}\label{ProjectionPoint}
\Phi_k\left({\bm{\eta}}^{(r)}\right) = \left(\eta_k^{(r)} - \underline \eta_k\right)\frac{1-{\sum\nolimits_{k = 0}^K} \underline \eta_k}{ {\sum\nolimits_{k = 0}^K}\eta_k^{(r)} - {\sum\nolimits_{k = 0}^K}\underline \eta_k} + \underline \eta_k,
	\vspace{-1.5mm}
\end{equation}
with $\underline \eta_k$ denoting the lower bound of the ratio of time slot $k$. A smaller polyblock ${\cal{Q}}^{(r+1)}$ can be constructed by replacing the vertices ${\bm{\eta}}^{(r)}$ with the newly generated vectors in polyblock ${\cal{Q}}^{(r)}$, and it converges until the upper bound of the achievable rate $\bar R^{(r)}$ approaches to the achievable rate of the projection point $\Phi\left({\bm{\eta}}^{(r)}\right)$. To reduce the search region, based on the constraints in (\ref{P1}b), a lower bound of $\eta_{k-1}$ can be given by 
\vspace{-2mm}
\begin{equation}\label{LowerBountEtaValue}
	\underline \eta_{k-1} \!= \! \gamma^{th} {{ \sigma_s^2} } \!\left( \!{W {P_{\mathrm{A}}}  \beta_{G,k,n}^2 L M_r h(\varphi_{k,n|n-1}, \sigma^2_{\omega_\varphi}) } \!\right)^{-1} \!\!.
	\vspace{-1mm}
\end{equation}
According to (\ref{LowerBountEtaValue}), an upper bound of $\eta_k$ is $\bar \eta_k = 1 - \sum\nolimits_{k' \ne k}^K \underline \eta_{k'}$. Then, (P2) can be solved more efficiently by searching $\eta_k$ within a smaller range $[\underline \eta_k, \bar \eta_k]$. The details of the proposed algorithm are provided in Algorithm \ref{MONObasedAlgorithm}.

\begin{algorithm}[t]
	\vspace{-0.2mm}
	\caption{Polyblock-Based Optimal Algorithm}
	\small
	\label{MONObasedAlgorithm}
	\begin{algorithmic}[1]
		\STATE Initialize the lower bound $\underline \eta_k$ and the upper bound $\bar \eta_k$ based on (\ref{LowerBountEtaValue}), put vector $ {\bm{\eta}}^{(r)} = [\bar \eta_0, \cdots, \bar \eta_K]$ into vertex set ${{\cal{W}}^{(r)}}$, and construct polyblock ${\cal{Q}}^{(r)}$ based on ${{\cal{W}}^{(r)}}$. Set $r = 0$, $\bar{R}^{(r)} = \infty$
		\STATE Let ${\rm{CBV}}^{(r)} = 0$ and set the convergence accuracy $\epsilon$				
		\WHILE {$\frac{|{\rm{CBV}}^{(r)}-\bar{R}^{(r)}|}{{\rm{CBV}}^{(r-1)}} > \epsilon$}
		\STATE Select the vertex ${\bm{\eta}}^{(r)}$ with the maximum achievable rate from ${\cal{W}}^{(r)}$. Calculate projection point $\Phi({\bm{\eta}}^{(r)})$ according to (\ref{ProjectionPoint})
		\STATE Obtain ${\rm{CBV}}^{(r+1)} = \mathop {\arg \min }\limits_k R_{k,n}$ under time allocation $\Phi({\bm{\eta}}^{(r)})$; construct a smaller polyblock ${\cal{Q}}^{(r+1)}$ by replacing vertices ${\bm{\eta}}^{(r)}$ in ${{\cal{W}}^{(r)}}$ 	
		\STATE Find vertex 
		${{ {\bm{\eta}}}^{(r+1)}} = \mathop {\arg \max }\limits_{{ {\bm{\eta}}} \in {{\cal{W}}^{(r+1)}}} \left\{ {R({ {\bm{\eta}}})} \right\}$, where $R({ {\bm{\eta}}})$ is the achievable rate under ${\bm{\eta}}$. The upper bound of achievable rate is $\bar{R}^{(r+1)} = R({{ {\bm{\eta}}}^{(r+1)}})$
		\STATE $r = r+1$
		\ENDWHILE
		\vspace{-0.5mm}
	\end{algorithmic}
\end{algorithm}

\vspace{-1mm}
\section{Simulations}
\vspace{-1mm}
\label{SimulationSection}
In this section, simulation results are provided for characterizing the performance of the proposed mutual assistance scheme. The system parameters are given as follows: $M_r \! =\! 10$, $L \!=\! 100$, $K \!=\! 3$, $\beta_0 \!=\! -30$dB, $\sigma^2_c \!=\! \sigma^2_s \!=\! -80$dB, $P_{\mathrm{A}} \!=\! 0.1$W, $\sigma^2_{\omega_{\varphi}} \!=\! 0.01$, $\sigma^2_R \!=\! 0.1$, $\epsilon \!= \!10^{-3}$, $\Delta T \!=\! 0.1$s, $W \!=\! 10^4$, and $\gamma^{th} \!=\! 10^3$. The RSU is located at (0m, 0m, 10m) and vehicles drive from ($-50$m, $-20$m, 0m) to (50m, $-20$m, 0m) at a speed of 15m/s. The proposed scheme is compared with four benchmarks as follows, and the corresponding benchmark protocols are shown in Fig.~\ref{figure1b}).
\vspace{-1mm}
\begin{itemize}[leftmargin=*]
	\item {\bf{Without sensing assistance (w/o s-assistance)}}: The phase shifts are designed based on the estimation angle $\varphi_{k,n|n-1}$.
	\item {\bf{Without communication assistance (w/o c-assistance)}}: Sensing results are used for phase shift design, but the S\&C tasks of vehicles are performed in a time-division manner.
	\item {\bf{Without S\&C assistance (w/o s\&c-assistance)}}: S\&C tasks are performed separately without mutual assistance.
	\item {\bf{Random phase shifts}}: Different from the proposed protocol, the IRS phase shifts are designed randomly.
\end{itemize}
\vspace{-0.5mm}
In Fig.~\ref{figure2}, the SNR of echo signals of Monte Carlo simulation approaches to that of the closed-form expression derived in Proposition 1, even for $L=50$, which confirms its tightness
in practical setups. In Fig.~\ref{figure3}, it is observed that the achievable rate obtained by all schemes decreases as the sensing requirement $\gamma^{th}$ increases, which is indeed expected since lager $\gamma^{th}$ values impose tighter constraints on the time allocation for communication. It can be seen from Fig.~\ref{figure3} that the achievable rate of the proposed scheme is significantly higher than that of the w/o s-assistance scheme, especially when the sensing requirement $\gamma^{th}$ is lower, since the proposed scheme can take the optimal time allocation to enlarge the communication gain brought by sensing. Moreover, under a given achievable rate of 4.9 bps/Hz, the sensing performance of the proposed scheme is significantly higher than that of the w/o c-assistance scheme. The main reason is that more information signals are utilized to facilitate localization, thereby reducing the time consumption for dedicated sensing. 

Fig.~\ref{figure4} shows that the achievable rate of the proposed scheme is significantly larger compared to the w/o s-assistance scheme when the maximum transmit power is larger. This is because a larger transmit power used for dedicated sensing can achieve a higher communication gain brought by sensing. It is observed from Fig.~\ref{figure4} that with a given communication rate, e.g., 3.3 bps/Hz, the proposed scheme can reduce power consumption by half as compared to the w/o c-assistance scheme by exploiting the mutual assistance between S\&C. Moreover, the achievable rate of the proposed scheme is improved significantly compared to the w/o c-assistance scheme when the maximum transmit power is smaller. The main reason is that for the w/o c-assistance scheme, more time needs to be allocated for dedicated sensing since the signals transmitted to other vehicles are not utilized for sensing. 
\begin{figure*}[!t]
	\begin{minipage}[t]{0.33\linewidth}
		\centering
		\includegraphics[width=5.8cm]{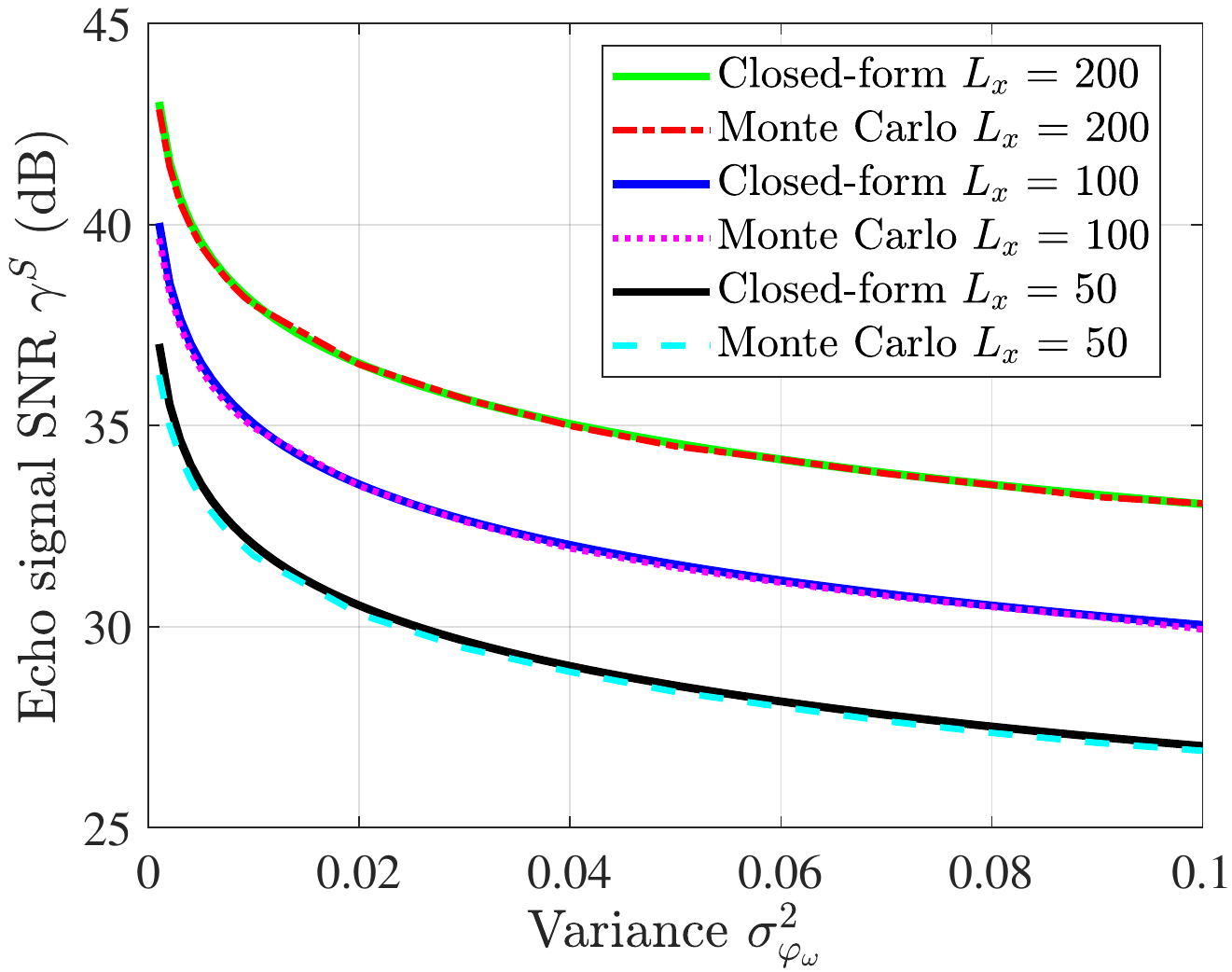}
		\vspace{-3mm}
		\caption{Evaluation of the derived expression.}
		\label{figure2}
	\end{minipage}%
	\begin{minipage}[t]{0.33\linewidth}
		\centering
		\includegraphics[width=5.8cm]{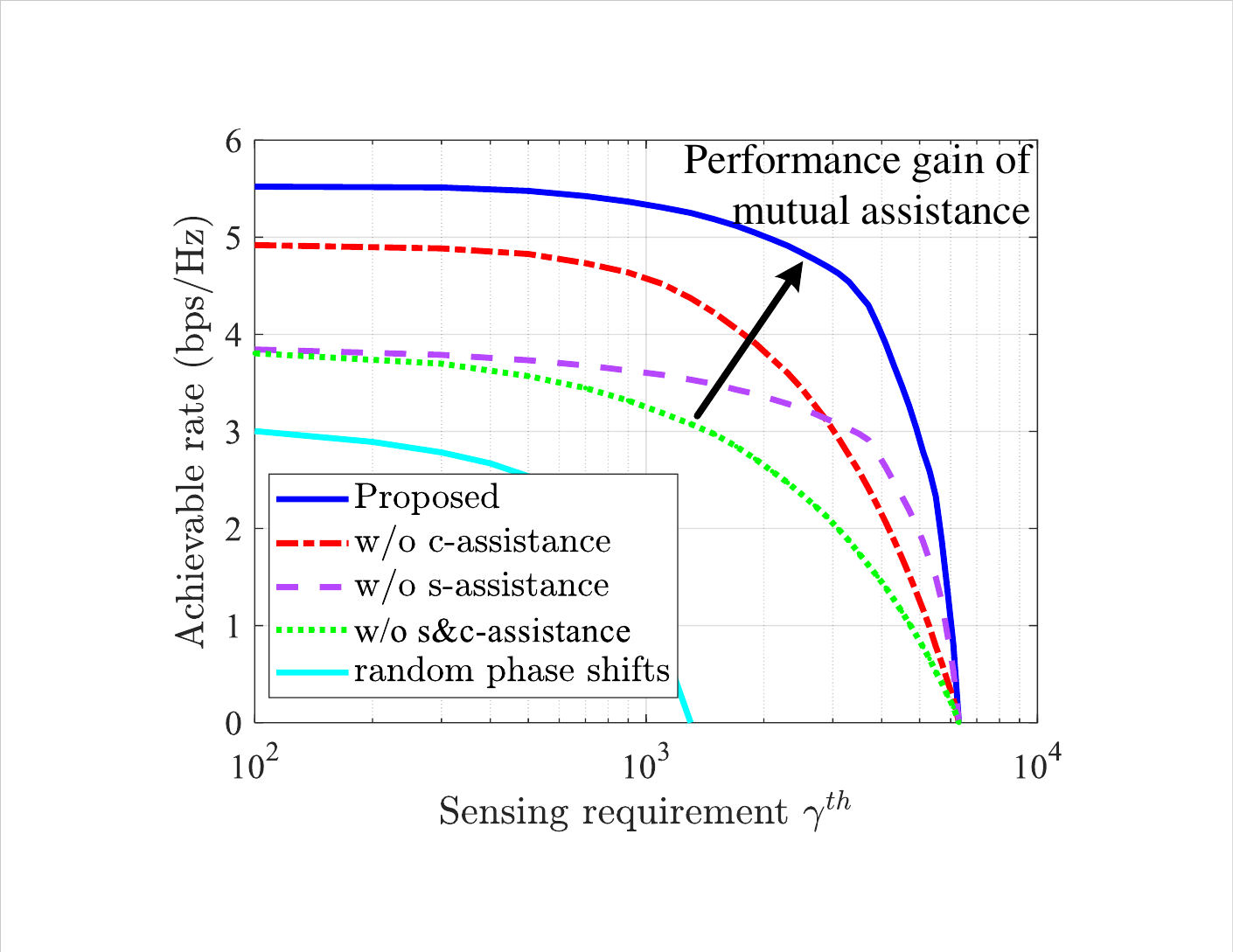}
		\vspace{-3mm}
		\caption{Achievable rate versus sensing requirement.}
		\label{figure3}
	\end{minipage}\hspace{1.5mm}
	\begin{minipage}[t]{0.33\linewidth}
		\centering
		\includegraphics[width=5.8cm]{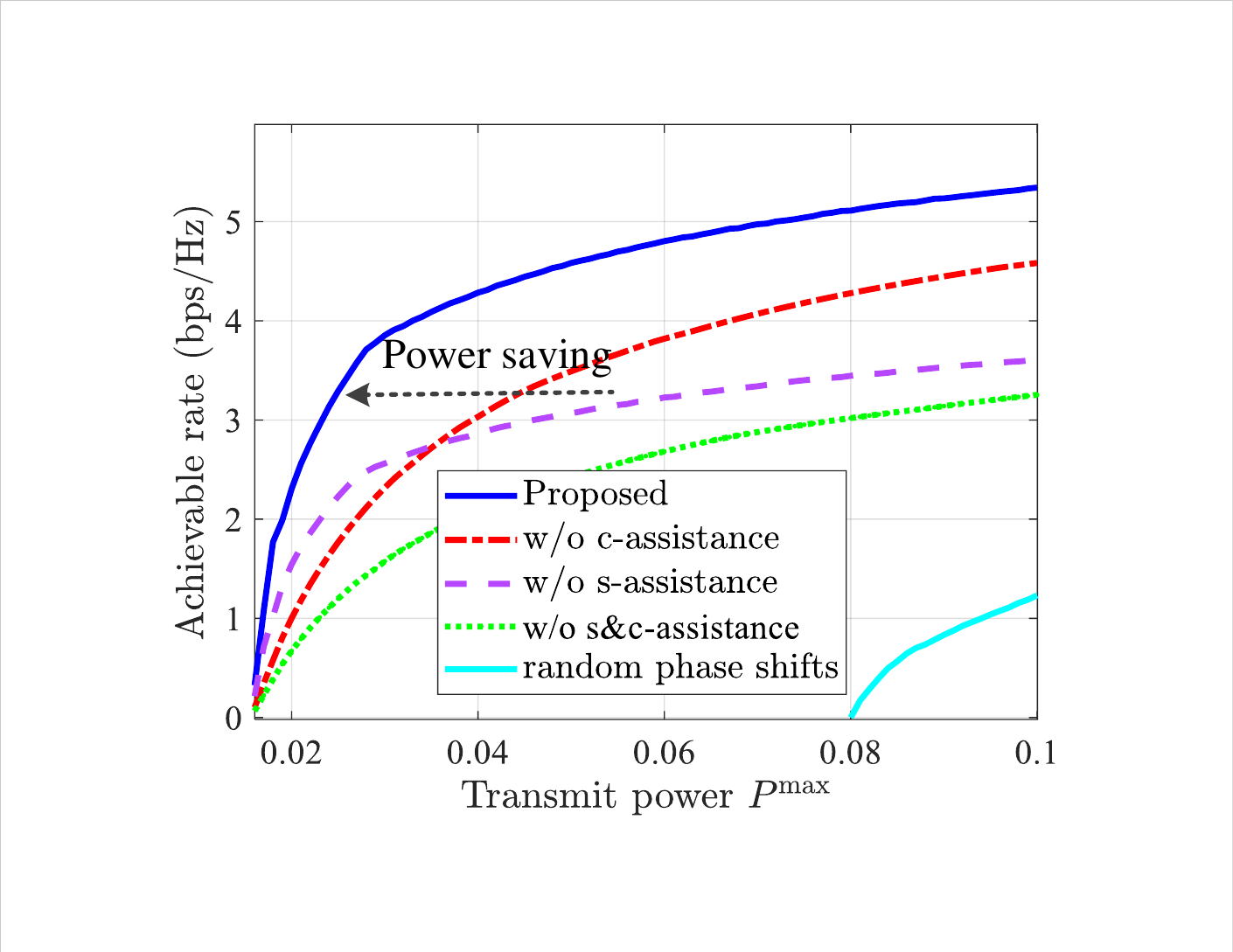}
		\vspace{-3mm}
		\caption{Achievable rate versus transmit power.}
		\label{figure4}
	\end{minipage}
\vspace{-4mm}
\end{figure*}

\vspace{-2mm}
\section{Conclusions and Future Works}
In this work, a novel mutual assistance scheme for S\&C was proposed to exploit coordination gain between S\&C in IRS-aided vehicular networks, where IRSs are set into reflecting/refracting modes for sensing/communication. First, a closed-form expression of the achievable rate was derived under uncertain angle information. Then, the search region of the proposed algorithm is reduced to facilitate the solution of the problem. Finally, simulation results demonstrated that the S\&C trade-off region was effectively enlarged through the mutual assistance design in IRS-aided vehicular networks. The scenarios with the multiple transmit antennas RSU considering the road of arbitrary geometry are worthwhile future works.

\begin{figure*}[!b]
	\vspace{-2mm}
	\begin{align}\label{ProbablityEquation}
		P_\varphi(y) \!= \! \frac{1}{{2{\sqrt {2\pi } {\sigma_{\omega_\varphi}}}\sqrt{ {1 \!-\! {{\left( {{\frac{y}{2}} \!+ \cos ( {{\varphi_{k,n}}} )} \right)}^2}} }}} \!\sum\limits_{i = -\infty}^{\infty} \!\!\left(\! {{e^{ \!- \frac{{{{\left( {2(i + 1)\pi  -  {\arccos ( {{\frac{y}{2}} + \cos ( {{\varphi_{k,n}}} )} )}  - {\varphi_{k,n}}} \right)}^2}}}{{2{\sigma^2_{\omega_\varphi}}}}}} \!+\! {e^{ \!- \frac{{{{\left( {2i\pi  + \arccos ( {y + \cos ( {{\varphi_{k,n}}} )} )} - \varphi_{k,n} \right)}^2}}}{{2{\sigma^2_{\omega_\varphi}}}}}}} \! \right)\!.
	\end{align}
\end{figure*}

\normalsize 
\vspace{-2mm}
\section*{Appendix A: \textsc{Proof of Proposition \ref{InftyBand}}}
\vspace{-0.5mm}
According to Theorem 3.6 in \cite{rust2013convergence}, if $g(x)$ is a real valued, continuous function with period $2$, it follows that 
\vspace{-1.5mm}
\begin{equation}\label{EquationLemmaFejer}
	\lim _{L \rightarrow \infty} \frac{1}{2} \int_{-1}^1 g(u)  F_{L}(u) d u \rightarrow g(0).
	\vspace{-1.5mm}
\end{equation}
Let $y \!=\! 2\cos ( {{ \varphi_{k,n|n-1}}} ) - 2\cos ( \varphi_{k,n}  )$, where ${ \varphi_{k,n|n-1}} = { \varphi_{k,n}} + \omega_\varphi$ and $\omega_\varphi \in \mathcal{C} \mathcal{N}(0, {\sigma^2_{\omega_\varphi}} )$. Then, the probability density function (PDF) of $y$ is given in (\ref{ProbablityEquation}), shown at the bottom of this page. Accordingly, we have
\vspace{-1.5mm}
\begin{align}\label{EquationProbability}
	&\mathbb{E}_{\varphi_{k,n|n-1}}\left[ F_{L}( 2\Delta \cos \varphi_{k,n}) F_{M_r}(\Delta \cos \varphi_{k,n})\right]	\\
	\approx & \int\nolimits_{ - 1 - \cos {\varphi_{k,n}}}^{1 - \cos {\varphi_{k,n}}} {\left( F_{M_r}\left(\frac{y}{2}\right) P_\varphi(y)\right) F_{L}\left( y\right)dy} 
	\stackrel{(a)}= 2 M_r P_\varphi(0), \nonumber
	\vspace{-1mm}
\end{align}
where $F_{L}(x) = \frac{1}{L}\left(\frac{\sin \frac{L \pi x}{2 }}{\sin \frac{\pi x}{2 }}\right)^{2}$. In (\ref{EquationProbability}), the approximation holds since $F_{M_r}\left(\frac{y}{2}\right) P_\varphi(y) \approx 0$ for $y \notin [-1 - \cos(\varphi_{k,n}), 1- \cos(\varphi_{k,n})]$, and ($a$) holds based on (\ref{EquationLemmaFejer}). $\varphi_{k,n}$ can be practically approximated by $\varphi_{k,n|n-1}$ since it has negligible effect on the value of $P_\varphi(0)$. This thus completes the proof.

\normalsize 
\section*{Appendix B: \textsc{Proof of Proposition \ref{VarianceMonotonicity}}}
We only need to verify that $\tilde h(\varphi_{k,n|n-1}, \sigma^2_{\tilde \varphi_{k,n}})$ increases monotonically with $\eta_{k-1}$ since $\tilde R_{k,n}$  increases monotonically with $\tilde h(\varphi_{k,n|n-1}, \sigma^2_{\tilde \varphi_{k,n}})$. Let $f(\eta_{k-1}) = \sqrt{\eta_{k-1}} \left( {1 + {e^{ - \eta_{k-1}}}} \right)$. Proposition \ref{VarianceMonotonicity} holds if $f(\eta_{k-1})$ is a monotonically increasing function of $\eta_{k-1}$ since $\tilde h(\varphi_{k,n|n-1}, \sigma^2_{\tilde \varphi_{k,n}})$ is an affine transformation of $f(\eta_{k-1})$. The derivative of $f(\eta_{k-1})$ satisfies  $f'(\eta_{k-1}) = \frac{1}{{\sqrt \eta_{k-1} }}\left( {\frac{1}{2} + \frac{1}{{{e^{\frac{1}{2}}}}}\left( {\frac{1}{2} - \eta_{k-1}} \right){e^{\frac{1}{2} - \eta_{k-1}}}} \right)  \stackrel{(b)}{\geqslant} \frac{1}{{\sqrt \eta_{k-1} }}\left( {\frac{1}{2} - \frac{1}{{{e^{\frac{1}{2}}}}}\frac{1}{e}} \right) > 0$, where ($b$) holds due to the inequality $\eta_{k-1}e^{\eta_{k-1}} \ge -\frac{1}{e}$ \cite{corless1996lambertw}. Also, $\tilde R_{k,n}$ increases monotonically with $\eta_k$. This thus completes the proof.

\footnotesize  	
\bibliography{mybibfile}
\bibliographystyle{IEEEtran}

\end{document}